\documentclass{article}

\usepackage{mathtools} % amsmath with fixes and additions
\usepackage{booktabs} % commands to create good-looking tables
\usepackage{tikz} % nice language for creating drawings and diagrams

\usepackage{amsmath,amsfonts,amssymb}
\usepackage{amsthm}

\usepackage[utf8]{inputenc} % allow utf-8 input
\usepackage[T1]{fontenc}    % use 8-bit T1 fonts
\usepackage{hyperref}       % hyperlinks
\usepackage{url}            % simple URL typesetting
\usepackage{booktabs}       % professional-quality tables
\usepackage{amsfonts}       % blackboard math symbols
\usepackage{nicefrac}       % compact symbols for 1/2, etc.
\usepackage{microtype}      % microtypography
\usepackage{amsmath,amsfonts,amssymb}
\usepackage{amsthm}
\usepackage{graphicx}
\usepackage{framed}
\usepackage{subfig}
\usepackage{commath}
\usepackage{authblk}
\usepackage{enumerate}

\usepackage{float}
\newfloat{Algorithm}{hbtp}{lop}[section]
\usepackage{boxedminipage}

\usepackage{xfrac}

\newcommand{\lv}[1]{}

\newcommand{\etal}{\emph{et~al.}}

\newcommand{\veps}{\varepsilon}

\newcommand{\X}{\mathbb{X}}

\newcommand{\R}{\mathbb{R}}

\newcommand{\N}{\mathbb{N}}
\newcommand{\Pa}{\mathbb{P}}
\newcommand{\Sp}{\mathbb{S}^{d-1}}

\newcommand{\Lp}[1]{L(X, k, #1)}

\newtheorem{definition}{Definition}
\newtheorem{lemma}{Lemma}
\newtheorem{theorem}{Theorem}

\usepackage{ifthen}
\newboolean{shortver}
\setboolean{shortver}{false}% for short version 
\graphicspath{{./Figures/}}

\title{On the $k$-Means/$k$-Median Cost Function}

\author[1]{Anup Bhattacharya}
\affil[1]{School of Computer Sciences, NISER}
\affil[]{\textit{bhattacharya.anup@gmail.com}}

\author[2]{Yoav Freund}
\affil[2]{CSE, UCSD}
\affil[]{\textit{yfreund@eng.ucsd.edu}}

\author[3]{Ragesh Jaiswal}
\affil[3]{CSE, IIT Delhi}
\affil[]{\textit{rjaiswal@cse.iitd.ac.in}}

\begin{document}
\maketitle

\begin{abstract}
In this work, we study the $k$-means cost function. Given a dataset $X \subseteq \R^d$ and an integer $k$, the goal of the Euclidean $k$-means problem is to find a set of $k$ centers $C \subseteq \R^d$ such that $\Phi(C, X) \equiv \sum_{x \in X} \min_{c \in C} \norm{x - c}^2$ is minimized. Let $\Delta(X,k) \equiv \min_{C \subseteq \R^d} \Phi(C, X)$ denote the cost of the optimal $k$-means solution. For any dataset $X$, $\Delta(X,k)$ decreases as $k$ increases. In this work, we try to understand this behaviour more precisely. For any dataset $X \subseteq \R^d$, integer $k \geq 1$, and a precision parameter $\veps > 0$, let $\Lp{\veps}$ denote the smallest integer such that $\Delta(X,\Lp{\veps}) \leq \veps \cdot \Delta(X,k)$. We show upper and lower bounds on this quantity. Our techniques generalize for the metric $k$-median problem in arbitrary metric spaces and we give bounds in terms of the {\em doubling dimension} of the metric. Finally, we observe that for any dataset $X$, we can compute a set $S$ of size $O \left(\Lp{\veps/c} \right)$ using {\em $D^2$-sampling} such that $\Phi(S,X) \leq \veps \cdot \Delta(X,k)$ for some fixed constant $c$. Next, we mention some applications of our bounds. First, analysing the pseudo-approximation guarantees of $k$-means++ seeding has been a popular research topic. Our results may be seen as non-trivial addition to the current state of knowledge. Secondly, our bounds imply that any constant approximation algorithm when executed with number of clusters $O \left(\Lp{\veps^2/c} \right)$ gives an {\em $(k, \veps)$-coreset} for the $k$-means problem. This implies that a $D^2$-sampled set of size $O \left(\Lp{\veps^2/c} \right)$ is a $(k, \veps)$-coreset. This is an improvement over similar results of Ackermann \etal\cite{streamkm}. Third, our results also imply that the rate of decrease of $\Delta_k(X)$ with $k$ depends on the {\it intrinsic dimension} of any dataset $X$. Hence the rate at which $\Delta_k(X)$ diminishes may be used to infer the intrinsic dimension of a dataset $X$. We propose a sampling based intrinsic dimension estimator and evaluate it on real and synthetic datasets.
\end{abstract}

\section{Introduction}\label{sec:intro}
The Euclidean $k$-means problem is one of the most well-studied problems in the clustering literature. 
The problem is defined in the following manner: 

\begin{definition}[$k$-Means problem] Given a dataset $X \subseteq \R^d$ and a positive integer $k$, find a set of $k$ points $C \subseteq \R^d$ (called {\em centers}) such that the cost function $\Phi(C, X) \equiv \sum_{x \in X} \min_{c \in C} D(x, c)^2$ is minimized, where $D(x, c) \equiv \norm{x - c}$. \end{definition}

In the weighted version of the $k$-means problem, there is a weight function $w: X \rightarrow \R^{+}$ and the cost function for the weighted $k$-means problem is defined as $\Phi(C, X, w) \equiv \sum_{x\in X} \min_{c \in C} \left(w(x) \cdot D(x, c)^2 \right)$. Let $\Delta(X,k)$ denote the optimal cost of the $k$-means objective function. That is $\Delta(X,k) \equiv \min_{C \subseteq \R^d,|C|=k} \Phi(C, X)$.
In this work, we try to understand the behaviour of $\Delta(X,k)$ as $k$ increases. 
More specifically, for a small precision parameter $\veps > 0$, we ask: what is the smallest integer $k'$ such that $\Delta(X,k')$ is at most $\veps \cdot \Delta(X,k)$? 
Note that when $\veps = 1$, $k' = k$ and as $\veps$ becomes smaller, $k'$ should grow. 
We are interested in understanding the relationship of $k'$ with input parameters such as the size of the dataset $n$, dimension $d$, and $k$. Next, we formally define the quantity $\Lp{\veps}$ for which we obtain bounds in this paper.

\begin{definition} For any dataset $X \subseteq \R^d$, precision parameter $0 < \veps \leq 1$ and positive integer $k$, let $\Lp{\veps}$ denote the smallest integer such that $\Delta(X,(\Lp{\veps}) \leq \veps \cdot \Delta(X,k)$. \end{definition}

We give upper and lower bounds on $\Lp{\veps}$ in terms of the geometric quantities known as the covering and packing numbers (\cite{matousek2002}). These are defined below.

\begin{definition}[Covering number] Let $(\X, D)$ be a metric space and let $0 < \veps \leq 1$. A subset $S$ of $\X$ is said to be an $\veps$-covering set for $\X$ iff for every $x \in \X$, there exists an $s \in S$ such that $D(x, s) \leq \veps$. The minimum cardinality of an $\veps$-covering set of $\X$, if finite, is called the covering number of $\X$ (at scale $\veps$) and is denoted by $\N(\X, \veps)$. \end{definition}

\begin{definition}[Packing number] Let $(\X, D)$ be a metric space and let $0 < \veps \leq 1$. A subset $S$ of $\X$ is said to be an $\veps$-packing set iff for every $x, y \in S$ such that $x \neq y$, we have $D(x, y) \geq \veps$. The maximum cardinality of an $\veps$-packing set of $\X$, if finite, is called the packing number of $\X$ (at scale $\veps$) and is denoted by $\Pa(\X, \veps)$. \end{definition}

Our bounds on $\Lp{\veps}$ are in terms of $\N(\Sp, \veps)$ and $\Pa(\Sp, \veps)$, where $\Sp$ denotes a unit sphere in $\R^d$. 
%Given this, it will be useful to know the bounds on $\N(\Sp, \veps)$ and $\Pa(\Sp, \veps)$. The proof of the next lemma may be found in the Appendix. 
The next lemma gives the bounds on $\N(\Sp, \veps)$ and $\Pa(\Sp, \veps)$. The proof of the lemma may be found in Appendix \ref{sec:appen}.

\begin{lemma}[Bounds on $\N(\Sp, \veps)$ and $\Pa(\Sp, \veps)$]\label{lemma:geometric-bounds}
Let $\Sp$ denote a unit sphere in $\R^d$.
%and $0 < \veps < 1/4$. 
Then,
%\begin{enumerate}
%\item $\N(\Sp, \veps) = O\left(\frac{1}{(\veps/8)^{d-1}}\right)$, and
%\item $\Pa(\Sp, \veps) = \Omega\left(\frac{1}{(2\veps)^{d-1}}\right)$.
%\end{enumerate}
\[\N(\Sp, \veps) = O\left(\frac{1}{(\veps/8)^{d-1}}\right), \Pa(\Sp, \veps) = \Omega\left(\frac{1}{(2\veps)^{d-1}}\right)\]

\end{lemma}

Here is our main result for the Euclidean $k$-means problem.

\begin{theorem}[Main result for $k$-means]\label{thm:main1}
Let $\Sp$ denote a unit sphere in $\R^d$.
The following holds for any $0 < \veps \leq 1/8$ and any positive integer $k$:
\begin{enumerate}
\item For any dataset $X \subseteq \R^d$ with $n$ points, $\Lp{\veps} = O \left( \frac{\N(\Sp, \sqrt{\frac{\veps}{2}}) \cdot k \cdot \log{n}}{\sqrt{\veps}}\right)$, and

\item There exists a dataset $X \subseteq \R^d$ with $n$ points such that $\Lp{\veps} = \Omega \left( \frac{\Pa(\Sp, \sqrt{8 \veps}) \cdot k \cdot \log{n}}{\sqrt{\veps}} \right)$.
\end{enumerate}
\end{theorem}

Note that a worse upper bound of $O\left( (9d/\veps)^{d/2} \cdot k \cdot \log{n}\right)$ for the Euclidean $k$-means problem was implicit in the work of \cite{streamkm} where as our bound is $O\left((128/\veps)^{(d-1)/2} \cdot k \cdot \log n \right)$. On the lower bound side, this question was open. Next, we show similar bounds for the metric $k$-median problem over arbitrary metrics. We first define the metric $k$-median problem over any metric $(\X, D)$.

\begin{definition}[Metric $k$-Median problem] Let $(\X, D)$ be any metric space. Given $X \subseteq \X$ and an integer $k$, find a set $C \subseteq \X$ of $k$ centers such that the cost function $\Phi(C, X) \equiv \sum_{x \in X} \min_{c \in C} D(x, c)$ is minimised. \end{definition}

We will use $\Delta(X,k)$ to denote the optimal cost of the metric $k$-median problem on dataset $X$ and $\Lp{\veps}$ to denote the smallest integer such that $\Delta(X,\Lp{\veps}) \leq \veps \cdot \Delta(X,k)$. 
However, here we obtain the bounds in terms of the {\em doubling dimension} of the metric. 
Let us first define the doubling dimension. 
The diameter $dia(X)$ of any set $X \subseteq \X$ is defined as $dia(X) = \max_{x, x' \in X}{D(x, x')}$. Given any set $X \subseteq \X$ and $r\in \R^{+}$, a set $\{X_1,X_2,\ldots,X_m\}$ is said to be an $r$-cover of $X$ iff $\cup_i X_i = X$ and for all $1 \leq i \leq m, dia(X_i) \leq r$. Given $X \subseteq \X$ and $r \in \R^{+}$, the covering number of the set $X$ with respect to diameter $r$, denoted by $\lambda(X, r)$, is the size of the $r$-cover of smallest cardinality. 
We can now define the doubling dimension of any metric $(\X, D)$.

\begin{definition}[Doubling dimension]
The doubling dimension of any metric $(\X, D)$ is the smallest integer $d$ such that for every $X \subseteq \X$, $\lambda \left(X, dia(X)/2 \right) \leq 2^d$.
\end{definition}

We obtain the above bounds for the metric $k$-median problem in terms of the doubling dimension. Here are the statements of upper and lower bounds that we obtain:

\begin{theorem}[Upper bound for metric $k$-median]\label{thm:metric-median}
Let $(\X, D)$ be any metric space with doubling dimension $d$.
For any $0 < \veps \leq 1$, any integer $k\geq 1$, and any dataset $X \subseteq \X$ with $n$ points, there exists a set $\xi \subseteq \X$ of size $O \left( \frac{k \cdot \log{n}}{(\veps/8)^d}\right)$ such that $\Phi(\xi, X) \leq \veps \cdot \Delta(X,k)$.
\end{theorem}

\begin{theorem}[Lower bound for metric $k$-median]
For any $0 < \veps \leq 1/8$, any integer $k\geq 1$, there exists a metric space $(\X, D)$ with doubling dimension $d$ and a dataset $X \subseteq \X$ with $n$ points, such that for any set $\xi \subseteq \X$ with $\Phi(\xi, X) \leq \veps \cdot \Delta(X,k)$, $\xi$ is of size $\Omega \left( \frac{k \cdot \log{n}}{(16\veps)^{d-1}}\right)$.
\end{theorem}

Next, we discuss a few applications of our bounds.

\subsection{Applications and related work}

The main applications of our bounds are in understanding the pseudo-approximation behaviour of the $k$-means++ seeding algorithm and coreset constructions for the $k$-means/$k$-median clustering problems. Here, we discuss mainly in terms of the (Euclidean) $k$-means problem, but most of the ideas may be extended for the $k$-median problem in any arbitrary metric. 
%We also discuss a sampling based intrinsic dimension estimation method motivated from our bounds.
%Another interesting application of our bounds is to estimate the intrinsic dimension of a dataset. 
%We discuss these applications as follows.

\subsubsection{Pseudo-approximation of $k$-means++}
$k$-means++ seeding is a sampling procedure that is popularly used as a seeding algorithm for the Lloyd's algorithm for $k$-means. The algorithm is given as follows.

\begin{quote}
    
({\bf $k$-means++ seeding or $D^2$-sampling}): Let $X \subseteq \R^d$. Pick the first center uniformly at random from $X$. After having picked $(i-1)$ centers denoted by $C_{i-1}$, pick a point $x \in X$ to be the $i^{th}$ center with probability proportional to $\min_{c \in C_{i-1}} D(x, c)^2$, where  $D(x, c) \equiv \norm{x - c}$.
\end{quote}

%The above sampling procedure is known to give good centers and has been a popular choice for picking the initial centers for the k-means algorithm in practice. Moreover, 
\cite{ArthurV07} showed that this algorithm gives an $O(\log{k})$-approximation guarantee in expectation. A lot of follow-up research has been done to understand the pseudo-approximation behaviour of this algorithm. 
%Note that $k$-means++ seeding uses $k$ only as a termination condition. 
Note that $k$-means++ seeding stops after sampling $k$ centers using $D^2$-sampling\footnote{Given a set $C$ of centers, $D^2$-sampling with respect to $C$ chooses point $x$ with probability proportional to $\min_{c\in C} D(x,c)^2$.}. If one continues to sample centers even after sampling $k$ of them, then do the sampled centers give better than $O(\log{k})$ pseudo-approximation? Pseudo-approximation means that the cost is calculated with respect to the sampled centers, more than $k$ in number, but compared with the optimal solution for $k$ centers. \cite{AggarwalDK09} analysed this behaviour and showed that if one samples $O(k)$ centers, then we get a constant factor pseudo-approximation. \cite{W16} showed that if $\beta k$ centers are sampled for {\em any} constant $\beta > 1$, then we get a constant factor pseudo-approximation in expectation. In a more recent work, \cite{MRS2020} gave improved pseudo-approximation guarantees for $k$-means++ when $k+\Delta$ centers are sampled using $D^2$-sampling, for some $\Delta>0$. Clearly, as the number of centers sampled using $D^2$-sampling increases, the $k$-means cost with respect to the sampled centers will decrease. Let us try to understand this behaviour. One way to formalise this is to find bounds on the number of samples such that the cost becomes at most $\veps$ times the optimal cost with respect to $k$ centers for any $0 < \veps \leq 1$. Combining our bounds with the results of \cite{AggarwalDK09} (i.e, $O(k)$ samples give constant approximation with high probability, we get the following.

\begin{theorem}\label{thm:pseudo}
        There is a universal constant $c$ for which the following holds: For any $0 < \veps \leq 1$, positive integer $k$, and any dataset $X \subseteq \R^d$, let $S$ denote a set of centers sampled with $D^2$-sampling such that $|S| = \Omega(\Lp{\veps/c})$. Then, with high probability, $\Phi(S, X) \leq \veps \cdot \Delta(X,k)$.
\end{theorem}

\subsubsection{Coresets for $k$-means}
Coresets are extremely useful objects in data processing, where a coreset of a large dataset can be thought of as a concise representation of the dataset with respect to the specific data processing task in question. Next, we give the formal definition of a coreset for the $k$-means problem.

\begin{definition}[$(k,\veps)$-coreset] A $(k,\veps)$-coreset of a set $X \subseteq \R^d$ is a set $S \subseteq \R^d$ along with a weight function $w: S \rightarrow \R^{+}$ such that for any set of $k$ centers $C \subseteq \R^d$, we have:
\[
(1 - \veps) \cdot \Phi(C, X) \leq \Phi(C, S, w) \leq (1 + \veps) \cdot \Phi(C, X)
\]
\end{definition}

A lot of work \cite{BadoiuHI02, Har-PeledM04, PeledKushal05, ls10, fl11, fss13} has been done in constructing coresets of small size.\footnote{The size of a coreset is the size of the set $S$ in the definition.} \cite{Har-PeledM04} and \cite{PeledKushal05} had coreset constructions by {\em quantization} of the space and finding points that may ``represent'' more than one point of the given dataset. In some sense, these coreset constructions are more geometric in nature than the more advanced constructions (see \cite{fss13}) and hence, these coresets are also known as ``movement-based'' coresets in \cite{mcc20}.
%It would be a good idea to abstract out some of the core ideas of these movement-based coreset constructions in the form of the next definition.
We define the notion of a movement-based coreset as follows.

\begin{definition}[$(k, \veps)$-movement-based coreset]
A $(k, \veps)$-movement-based coreset of a dataset $X \subseteq \R^d$ is a set of points $S \subseteq \R^d$ such that:
%for any set of $k$ centers $C \subseteq \R^d$, we have:
\[
\Phi(S, X) \leq \veps \cdot \Delta(X,k).
\]
\end{definition}

%We call such a coreset a {\em geometric coreset}. 
We will show that sampling $O(\Lp{\veps^2/\beta})$ points using $D^2$-sampling gives a $(k,\veps)$-movement-based coreset for some constant $\beta$.

\subsubsection{Estimation of of intrinsic dimension}

The intrinsic dimension of a dataset may be thought of as the minimum number of parameters required to account for the observed properties of the dataset. For most datasets, the {\em extrinsic dimension} (observed dimension) is much larger than the intrinsic dimension. The performance of many data analysis algorithms deteriorates as the dimension increase (popularly known as the curse of dimensionality). However, in many contexts, the data actually lies in a much lower dimensional space. In that case the intrinsic dimension of the dataset is much lower than the extrinsic dimension. The easiest example to consider is a dataset $X \subseteq \R^d$ such that the $X$ sits in a $D \ll d$ dimensional subspace of $\R^d$. In this case the intrinsic dimension is $D$. 

Since one of the key algorithmic tools to tackle the high dimensional data is dimensionality reduction, estimating the intrinsic dimension of a given dataset becomes an important task. 

There has been a lot of work in developing techniques for estimation of intrinsic dimension of a dataset. \cite{CS16} gave a nice survey on this topic. There are a number of ways to formalize the above intuitive notion of intrinsic dimension. One way to formalize is to say that a high dimensional dataset with extrinsic dimension $d$ has intrinsic dimension $D < d$ if the data sits within a $D$-dimensional sub-manifold. \cite{RL06} gave a dimension estimation technique under the assumption that the data points are uniformly distributed on a $D$-dimensional compact smooth sub-manifold of $\mathbb{R}^{d}$. Their technique involves finding the rate at which the {\em quantization error} diminishes with respect to the rate of the quantizer. Since some of the above terms have not been defined in our current context, let us rephrase them in the context of the $k$-means problem. Essentially, they estimated the intrinsic dimension to be the slope of the log-log plot of $k$ versus $(\Delta(X,k))^{1/2}$. The theoretical justification for such an estimation is that for any regular probability measure over a $D$-dimensional compact manifold, the expectation of the quantization error (i.e., expectation of $(\Delta(X,k))^{1/2}$ for random variable $X$) behaves as $\Theta(k^{-1/D})$. However, such a result is not known for an arbitrary discrete distribution (or an arbitrary set of points). Our results about the behaviour of the $k$-means cost function provides justification for the same estimation method for any arbitrary discrete distribution or even an arbitrary set of points for certain notions of intrinsic dimension similar to that considered by Raginski and Lazebnik. We provide experimental results for estimation of intrinsic dimension on some common datasets.

\subsubsection*{Organization of the Paper} In Section \ref{sec:bounds} and \ref{sec:kmedian}, we prove the bounds for the $k$-means and $k$-median problems. Section \ref{sec:coresets} details the construction of a $k$-means coreset. In Section \ref{sec:intrinsic}, the technique for estimation of intrinsic dimension of a dataset is described and experimental results are given.

\section{Bounds for Euclidean $k$-means}\label{sec:bounds}
In this section, we prove the bounds on $\Lp{\veps}$. We do this by using ideas in \cite{PeledKushal05} to reduce the high-dimensional case to a one-dimensional case. We start by discussing the upper and lower bounds for the one-dimensional data.
%In the discussion below, we will make use of the definition of the distance of a given point $x$ from a point set $S$, denoted by $D(x, S)$ and defined as $D(x, S) = \min_{s \in S}\norm{s - x}$.
Let the distance of a point $x$ from a set $S$ be denoted as $D(x,S)=\min_{s \in S}\norm{s - x}$.
%Also, note that the big-O notation used in the bounds below only hide fixed universal constants. 

\begin{lemma}[Upper bound]\label{lem:1}
For any $0 < \veps \leq 1$ and any dataset $X \subseteq \R$ with $n$ points such that it is centered at the origin, there exists a set $S \subseteq \R$ of size $O \left( \sfrac{\log{n}}{\sqrt{\veps}}\right)$ such that $\Phi(S, X) \leq \veps \cdot \Phi(\{0\}, X)$.
\end{lemma}

\begin{proof}
Let $R = \frac{\sqrt{\Phi(\{0\}, X)}}{n} = \frac{\sqrt{\sum_{x \in X} \norm{x}^2}}{n}$.
Let $r = (1 + \sqrt{\veps/2})$ and $t = \lceil \sfrac{\log{n}}{\log{r}} \rceil$.
Let 
\begin{eqnarray*}
\noindent S_1 &=& \{\pm (i \cdot \sqrt{\veps/2} \cdot R) ~|~ 0 \leq i \leq \lfloor \sqrt{2/\veps}\rfloor \} \\
\noindent S_2 &=& \{\pm (r^i \cdot R) ~|~ 0 \leq i \leq t\}
\end{eqnarray*}
We use $S = S_1 \cup S_2$.
Note that for our choice of $t$, there is no point $x \in X$ such that $\norm{x} > r^t \cdot R$.
For any point $x \in X$ such that $r^i R \leq \norm{x} < r^{i+1} R$, we have $D(x, S)^2 \leq (\veps/2) (r^i R)^2 \leq (\veps/2) \norm{x}^2$.
Also, note that for every $x$ such that $\norm{x} \leq R$, we have $D(x, S)^2 \leq (\veps/2) \cdot R^2$. 
So, we have:
\begin{eqnarray*}
&~~&\Phi(S, X)\\
&=& \sum_{x \in X} D(x, S)^2 \\
&=& \sum_{x \in X, \norm{x}\leq R} D(x, S)^2 + \sum_{x \in X, \norm{x} > R} D(x, S)^2\\
&\leq& \sum_{x \in X, \norm{x}\leq R} (\veps/2) R^2 + \sum_{x \in X, \norm{x} > R} (\veps/2) \norm{x}^2\\
&\leq& \veps \cdot \Phi(\{0\}, X).
\end{eqnarray*}
Finally, note that 
$$|S| = |S_1| + |S_2| = O \left( \frac{\log{n}}{\log{(1 + \sqrt{\veps/2})}}\right) = O \left( \frac{\log{n}}{\sqrt{\veps}}\right).$$ 
This completes the proof of the lemma.
\end{proof}

The lower bound instance for the 1-dimensional $k$-means problem below is based on ideas similar to the lower bound instance by \cite{bja}.

\begin{lemma}[Lower bound]\label{lem:2}
For any $0 < \veps < 1/8$, there exists a dataset $X \subseteq \R$ with $n$ points such that any set $S \subseteq \R$ with $\Phi(S, X) \leq \veps \cdot \Phi(\{0\}, X)$ satisfies $|S| = \Omega \left( \sfrac{\log{n}}{\sqrt{\veps}}\right)$.
\end{lemma}

\begin{proof}
In the dataset that we construct multiple points may be co-located.
Let $r = \lceil (1 + \sqrt{32 \veps}) \rceil$.
Consider the dataset $X$ described in the following manner: There are $r^{2(t-1)}$ points co-located at $\pm r$, $r^{2(t-2)}$ points co-located at $\pm r^2$, $r^{2(t-3)}$ points co-located at $\pm r^3$, ..., $1$ point located at $\pm r^{t}$.

We can fix the value of $t$ in the above description in terms of $n = |X|$ by noting that: $n = 2 \cdot (1 + r^{2}  + ... +  r^{2(t-1)})$.
%Given this, we can use $2t =  \left\lceil \frac{\log{ \left(\frac{n(r^2-1)}{2} + 1 \right)}}{\log{r}} \right\rceil$ in the description of the dataset.
Given this, we have $t = \Omega (\log_r{ ({n(r^2-1)}/{2} + 1)})$.
The cost with respect to single center at the origin is given by:
%\begin{eqnarray*}
$\Phi(\{0\}, X) = 2 \cdot (r^2 \cdot r^{2(t-1)} + r^4 \cdot r^{2(t-2)} + ... + r^{2t}) = 2t r^{2t}$.
%\end{eqnarray*}
Consider intervals around each of the populated locations. Let $I^{+}_{i} = [r^i(1 - \sqrt{2\veps}), r^i (1 + \sqrt{2\veps})]$ and $I^{-}_{i} = [-r^i(1 + \sqrt{2\veps}), -r^i (1 - \sqrt{2\veps})]$. Note that these intervals are disjoint for our choice of $r$. Consider any set $S$ with less than $t$ points. Note that there will be at least $t$ intervals that do not contain a point from the set $S$.
The points located at each of these ``uncovered'' intervals contribute a cost of at least $(2\veps) r^{2t}$. Given this, we have $\Phi(S, X) > t \cdot (2\veps) r^{2t} > \veps \cdot \Phi(\{0\}, X)$. So, for any set $S$ such that $\Phi(S, X) \leq \veps \cdot \Phi(\{0\}, X)$, we have $|S| > t$ which gives 
$|S| > t = \Omega(\sfrac{\log{n}}{\log{\lceil(1 + \sqrt{32 \veps}})\rceil})$ which gives the statement of the lemma.
\end{proof}

We will now extend these bounds to higher dimensions using the ideas of \cite{PeledKushal05}. We will use $\veps$-covering and $\veps$-packing numbers over the surface of unit spheres crucially in our construction. We now show the upper bound for points in $\R^d$.

\begin{theorem}[Upper bound for Euclidean $k$-means]\label{thm:1}
Let $\Sp$ denote the surface of the unit sphere in $\R^{d}$.
For any $0 < \veps \leq 1$, any integer $k\geq 1$, and any dataset $X \subseteq \R^d$ with $n$ points, there exists a set $\xi \subseteq \R^d$ of size $O \left( \frac{\N(\Sp, \sqrt{\veps/2}) \cdot k \cdot \log{n}}{\sqrt{\veps}}\right)$ such that $\Phi(\xi, X) \leq \veps \cdot \Delta(X,k)$.
\end{theorem}

\begin{proof} Let $C = \{c_1,\ldots,c_k\}$ denote the optimal centers for $k$-means on $X$ and let $X_1,\ldots,X_k$ be the optimal clusters. This also means that for all $i$, $c_i$ is the centroid of the point set $X_i$. It will be sufficient to find $k$ sets $\xi_1,\ldots,\xi_k$ such that for all $i$, $\Phi(\xi_i, X_i) \leq \veps \cdot \Phi(\{c_i\}, X_i)  = \veps \cdot \Delta({X_i},1)$. 
Given such sets $\xi_1,\ldots,\xi_k$, let $\xi = \cup_i \xi_i$. Then, we have $\Phi(\xi, X) \leq \sum_{i} \Phi(\xi_i, X_i) \leq \sum_i \veps \cdot \Delta({X_i},1) = \veps \cdot \Delta(X,k)$.

Consider any optimal cluster $X_i$. Let $Y = X_i$ and $c = c_i$. Let points in the set $Y$ be denoted as $y_1,\ldots,y_m$. We will now construct a set $S$ such that $\Phi(S, Y) \leq \veps \cdot \Phi(\{c\}, Y)$. %Recall that $\mathcal{R}(d, {\sqrt{\frac{\veps}{2}}})$ denotes an $\sqrt{\veps/2}$-net over the surface of a sphere of unit radius in $\R^d$.
Consider a unit sphere $\Sp$ around $c$ and let $R$ denote an $(\sqrt{{\veps}/{2}})$-covering set over $\Sp$ (note that the size of $R$ is $\N(\Sp, \sqrt{\veps/2})$).
This implies that for any point $r_1 \in \Sp$, there exists $r_2 \in R$ such that $r_2 \neq r_1$ and $\norm{r_1 - r_2} \leq \sqrt{\veps/2}$.
Consider a ``fan'' consisting of $|R|$ lines $F = \{l_1,\ldots, l_{|R|}\}$ connecting $c$ to each of the points in $R$. For any point $y \in Y$, let $y'$ denote the projection of $y$ on the nearest line among $l_1,\ldots,l_{|R|}$. Let $Y'$ denote the set of projected points. For any line $l \in \{l_1,\ldots,l_{|R|}\}$, let $Y_l$ denote the subset of projected points that are on line $l$.
We first observe that ``snapping'' the points to the fan $F$ does not cost much. This follows easily from the following simple observation. For all $y \in Y$, we have $\norm{y - y'} \leq \sqrt{\veps/2} \cdot \norm{y - c}$.

This implies that
\begin{equation}\label{eqn:1}
\mathcal{E} \equiv \sum_{y \in Y} \norm{y - y'}^2 \leq {\veps}/{2} \cdot \sum_{y \in Y} \norm{y - c}^2 = {\veps}/{2} \cdot \Delta(Y,1).
\end{equation}

We have:
\begin{equation}\label{eqn:2}
\Delta(Y,1) = \mathcal{E} + \sum_{l \in F} \Phi(\{c\}, Y_l)
\end{equation}
Now, from Lemma~\ref{lem:1}, we know that for every $l \in F$, there exists a set $S_l$ of $O \left(\sfrac{\log{|Y_l|}}{\sqrt{\veps}} \right)$ points on line $l$ such that $\Phi(S_l, Y_l) \leq (\veps/2) \cdot \Phi(\{c\}, Y_l)$.
Let $S = \cup_l S_l$. Then we have
\begin{eqnarray*}
&~~&\Phi(S, Y)\\
&=& \mathcal{E} + \sum_{y \in Y'} D(y, S)^2 \\
&\leq& \mathcal{E} + \sum_{l \in F} \sum_{y \in Y_l} D(y, S_l)^2 \\
&=& \mathcal{E} + \sum_{l \in F} \Phi(S_l, Y_l) \\
&\leq& \mathcal{E} + (\veps/2) \cdot \sum_{l \in F} \Phi(\{c\}, Y_l) \quad \textrm{(using Lemma~\ref{lem:1})}\\
&\leq& \mathcal{E} + (\veps/2) \cdot \Delta(Y,1) \quad \textrm{(using (\ref{eqn:2}))}\\
&\leq& (\veps/2) \cdot \Delta(Y,1) + (\veps/2) \cdot \Delta(Y,1) \quad \textrm{(using (\ref{eqn:1}))}\\
&=& \veps \cdot \Delta(Y,1)
\end{eqnarray*}

The size of the set $S$ is $O \left(\frac{\N(\Sp, \sqrt{\veps/2}) \cdot \log{|Y|}}{\sqrt{\veps}} \right)$. 
%Since the size of an $\alpha$-net on unit sphere in $\R^d$ is $O(1/\alpha^{d-1})$, we get that $|S| = O \left( \frac{\log{|Y|}}{(\veps/2)^{d/2}} \right)$. 
Repeating the same for all $k$ optimal clusters, we get a set $\xi$ of size $O \left( \frac{\N(\Sp, \sqrt{\veps/2})  \cdot k \cdot \log{n}}{\sqrt{\veps}}\right)$.
\end{proof}

Note that our upper bound is better than the bound of \cite{streamkm} and the improvement may be attributed to reducing the $d$-dimensional case to a $1$-dimensional case. We also give a lower bound below which essentially shows that the factors of $k$, $\log{n}$ and $(1/\veps)^{d/2}$ are unavoidable.

%\cite{streamkm} gave an upper bound of $O \left( (9d/\veps)^{d/2} \cdot k \cdot \log{n}\right)$. Our bound is better and the improvement may be attributed to reducing the $d$-dimensional case to a $1$-dimensional case. On the other hand, we also give a lower bound which is not done by \cite{streamkm}. Our lower bound essentially shows that the factors of $k$, $\log{n}$ and $(1/\veps)^{d/2}$ are unavoidable in some sense.

\begin{theorem}[Lower bound for Euclidean $k$-means]\label{thm:2}
For any $0 < \veps < 1/8$, any integer $k \geq 1$, there exists a dataset $X \subseteq \R^d$ with $n$ points such that any set $\xi \subseteq \R^d$ with $\Phi(\xi, X) \leq \veps \cdot \Delta(X,k)$ satisfies $|\xi| = \Omega \left( \frac{\Pa(\Sp, \sqrt{8 \veps}) \cdot k \cdot \log{n}}{\sqrt{\veps}} \right)$.
\end{theorem}

\begin{proof}
Let $c_1,\ldots,c_k \in \R^d$ be $k$ points such that $\forall i \neq j, \norm{c_i - c_j} > n^2$.
We will define $k$ sets of points $X_1,\ldots,X_k$, and our dataset will be $X = \cup_i X_i$.
Let $R$ denote an $\sqrt{8 \veps}$-packing set over a unit sphere in $\R^d$ (note that the size of $R$ is the packing number and is denoted by $\Pa(\Sp, \sqrt{8\veps})$).
For any $i$, here is the description of the set $X_i$: Let $\Sp(c_i)$ denote a unit sphere around $c_i$ and let $R(c_i)$ denote the $(\sqrt{8\veps})$-packing set laid over $\Sp(c_i)$. Consider a ``fan'' $F$ of lines connecting $c_i$ to each point in the set $R(c_i)$. Each line has $\eta = \sfrac{n}{(k \cdot \Pa(\Sp, \sqrt{8 \veps}))}$ points and these points are arranged as in the one-dimensional example of Lemma~\ref{lem:2}.

The analysis is very similar to that in Lemma~\ref{lem:2}. Instead of considering intervals around each populated location, we will consider balls of certain radius.
As in Lemma~\ref{lem:2}, we use $r = \lceil (1 + \sqrt{32 \veps}) \rceil$.
Also, $t = \Theta \left( \sfrac{\log{\eta}}{\log{r}}\right)$.
The populated locations are at distances $r, r^2,\ldots,r^t$ from the $c_i$'s.
We have $\Delta({X_i},1) = \Phi(\{c_i\}, X_i) = (2 t r^{2t}) \cdot \Pa(\Sp, \sqrt{8 \veps})$ which gives $\Delta(X,k) = \sum_i \Phi(\{c_i\}, X_i) = k \cdot (2 t r^{2t}) \cdot \Pa(\Sp, \sqrt{8 \veps})$.
Consider balls of radius $\sqrt{2 \veps} \cdot r^j$ around any populated location at a distance $r^j$ from $c_i$. Note that all of these balls are disjoint because of our choice of $r$ and because of the fact that the populated points are defined using an $\sqrt{8 \veps}$-packing set over unit sphere. The number of balls defined is $2kt \cdot \Pa(\Sp, \sqrt{8 \veps})$. Consider any set $S$ containing less than $kt \cdot \Pa(\Sp, \sqrt{8 \veps})$ points. There are at least $kt \cdot \Pa(\Sp, \sqrt{8 \veps})$ balls that do not contain any points from $S$. The cost contribution from the points located in each of these balls is at least $(2 \veps) r^{2t}$. So, $\Phi(S, X) > (2 \veps) \cdot kt \cdot \Pa(\Sp, \sqrt{8 \veps}) \cdot r^{2t} > \veps \cdot \Delta(X,k)$. Therefore, any set $\xi$ for which $\Phi(\xi, X) \leq \veps \cdot \Delta(X,k)$ satisfies 
\begin{align*}
|\xi| = & \Omega(kt \cdot \Pa(\Sp, \sqrt{8 \veps})) = \Omega \left( k \cdot \frac{\log{\eta}}{\log{r}} \cdot \Pa(\Sp, \sqrt{8 \veps})\right) \\
= & \Omega \left( \frac{\Pa(\Sp, \sqrt{8 \veps}) \cdot k \cdot \log{n}}{\sqrt{\veps}}\right)
\end{align*}
This completes the proof of the theorem.
\end{proof}

\subsection{Bounds for Euclidean $k$-median}
The Euclidean $k$-median problem is very similar to the Euclidean $k$-means problem except that the cost function is defined using ``sum of distances'' rather than ``sum of squared distances''. 
Given $X \subseteq \R^d$ and a positive integer $k$, find a set of centers $C \subseteq \R^d$ such that $\Phi(C, X) \stackrel{def}{=} \sum_{x \in X} \min_{c \in C} D(x, c)$ is minimized, where $D(x, y) = \norm{x - y}$. We obtain the following bounds, the details of which we omit here, by replacing the squared Euclidean distances with the Euclidean distances in all the above arguments.
%, we obtain the following bounds for the Euclidean $k$-median problem. We avoid giving the unnecessary details of the proof which is almost the same as that for $k$-means.

\begin{theorem}[Upper bound for Euclidean $k$-median]\label{thm:med1}
For any $0 < \veps \leq 1$, any integer $k\geq 1$, and any dataset $X \subseteq \R^d$ with $n$ points, there exists a set $\xi \subseteq \R^d$ of size $O \left( \frac{\N(\Sp, \veps/2) \cdot k \cdot \log{n}}{\veps}\right)$ such that $\Phi(\xi, X) \leq \veps \cdot \Delta(X,k)$.
\end{theorem}

\begin{theorem}[Lower bound for Euclidean $k$-median]\label{thm:med2}
For any $0 < \veps < 1/8$, any integer $k \geq 1$, there exists a dataset $X \subseteq \R^d$ with $n$ points such that any set $\xi \subseteq \R^d$ with $\Phi(\xi, X) \leq \veps \cdot \Delta(X,k)$ satisfies $|\xi| = \Omega \left( \frac{\Pa(\Sp, 4\veps) \cdot k \cdot \log{n}}{\veps} \right)$.
\end{theorem}

\section{Bounds for Metric $k$-Median}\label{sec:kmedian}
In this section we obtain bounds for $\Lp{\veps}$ for the metric $k$-median problem over any arbitrary metric space $(\X, D)$. Given $X \subseteq \X$ and a positive integer $k$, the metric $k$-median problem asks to find a set $C \subseteq \X$ of $k$ centers such that $\Phi(C, X) \equiv \sum_{x \in X} \min_{c \in C} D(x, c)$ is minimized.
%The metric $k$-median problem over arbitrary metric space $(\X, D)$ is defined as follows: Given $X \subseteq \X$ and a positive integer $k$, find a set $C \subseteq \X$ of $k$ centers such that $\Phi(C, X) \equiv \sum_{x \in X} \min_{c \in C} D(x, c)$ is minimized. 
%Let $\Delta(X,k) = \min_{C \subseteq \X} \Phi(C, X)$ denote the cost of the optimal clustering.
%As in the previous section, given any $X \subseteq \X$, let $\Lp{\veps}$ denote the smallest integer such that $\Delta(X,\Lp{\veps}) \leq \veps \cdot \Delta(X,k)$. We will show an upper bound very similar to the bound in the previous section. The important difference here is that the bounds in this section will be in terms of the {\em doubling dimension} of the metric. Let us recall the definition of doubling dimension.
We obtain bounds for $\Lp{\veps}$ in terms of the doubling dimension of the metric. Let us first recall the notion of doubling dimension. Given $X \subseteq \X$, the diameter of the set $X$, $dia(X)$, is defined as $dia(X) = \max_{x, x' \in X}{D(x, x')}$. Given any set $X \subseteq \X$ and $r \in \R^{+}$, a set $\{X_1, X_2,\ldots,X_m\}$ is said to be an $r$-cover of $X$ iff $\cup_i X_i = X$ and for all $1 \leq i \leq m, dia(X_i) \leq r$. Given $X \subseteq \X$ and $r \in \R^{+}$, the covering number of the set $X$ with respect to diameter $r$, denoted by $\lambda(X, r)$, is the size of the $r$-cover of smallest cardinality. We now define the doubling dimension of any metric $(\X, D)$.

\begin{definition}[Doubling dimension]
The doubling dimension of any metric $(\X, D)$ is the smallest integer $d$ such that for every $X \subseteq \X$, $\lambda \left(X, dia(X)/2 \right) \leq 2^d$.
\end{definition}

The remaining discussion will be with respect to the doubling dimension $d$ of any metric $(\X, D)$. We will use the following lemma for defining our upper bound.

\begin{definition}[$\veps$-covering number] Let $0 < \veps < 1$ be some precision parameter. For any $c \in \X$ and $r \in \R^{+}$, let $M(c, r) = \{x~|~x \in \X \textrm{ and } D(c, x) \leq r\}$. The $\veps$-covering number, denoted using $\gamma_{\veps}$, is defined as $\gamma_{\veps} = \max_{c \in \X, r \in \R^{+}}{(\lambda(M(c, r), \veps \cdot r))}$.
\end{definition}

In the next lemma, we give a bound on the $\veps$-covering number for any metric with doubling dimension $d$.

\begin{lemma}\label{lem:met4} Let $(\X, D)$ be any metric space. For any $0 < \veps < 1$, the $\veps$-covering number of $\X$, $\gamma_{\veps} = O \left((4/\veps)^d\right)$.
\end{lemma}
\begin{proof} The proof follows from the definition of $\lambda$. Consider any $c \in \X$ and any $r \in \R^{+}$. Let $X = M(c, r)$. From the triangle inequality we know that $dia(X) \leq 2r$. Using the fact that the doubling dimension of the metric is $d$, we get $\lambda(X, r) \leq 2^d, \lambda(X, r/2) \leq 2^{2d}$,  $\lambda(X, r/2^2) \leq 2^{3d}$ and so on. Let $\sfrac{1}{2^{k+1}} < \veps \leq \sfrac{1}{2^k}$. Then we have $\lambda(X, \veps \cdot r) \leq 2^{(k+2)d}$. This gives us $\lambda(X, \veps \cdot r) = O((4/\veps)^d)$.
\end{proof}

We can now give the upper bound in terms of the $\veps$-covering number of the metric.

\begin{theorem}[Upper bound for metric $k$-median]
Let $(\X, D)$ be any metric space with doubling dimension $d$.
For any $0 < \veps \leq 1$, any integer $k\geq 1$, and any dataset $X \subseteq \X$ with $n$ points, there exists a set $\xi \subseteq \X$ of size $O \left( \frac{k \cdot \log{n}}{(\veps/8)^d}\right)$ such that $\Phi(\xi, X) \leq \veps \cdot \Delta(X,k)$.
\end{theorem}
\begin{proof} Let $c_1,\ldots,c_k$ be the optimal $k$ centers and let $X_1,\ldots,X_k$ be the optimal clusters with respect to $c_1,\ldots,c_k$. It is sufficient to show that for all $i$, there exists a set $\xi_i \subseteq X_i$ of size $O \left(\sfrac{\gamma_{\veps} \cdot \log{n}}{\veps} \right)$ such that $\Phi(\xi_i, X_i) \leq \veps \cdot \Delta({X_i},1)$. Let $c = c_i$ and $Y = X_i$. Let $R = {\sum_{y \in Y} D(y, c)}/{|Y|}$ and $t = \Theta \left( \log_{2}{|Y|}\right)$. Consider ``concentric circles'' of radius $R, 2R, 2^2 R,\ldots, 2^t R$ around the center $c$. That is consider the sets $Y_0,\ldots, Y_t$ defined in the following manner.
\begin{eqnarray*}
Y_0 &=& \{y \in Y | D(c, y) \leq R\} \quad \textrm{and} \\ 
Y_{j} &=& \{y \in Y | D(c, y) \leq 2^j \cdot R\} \setminus Y_{j-1} \quad \textrm{for}~ 1\leq j\leq t
\end{eqnarray*}
First, note that $\cup_{j=0}^{t} Y_j = Y$ due to our choice of $t$.
Also, for every $j$, there exists subsets $S_j^{1}, S_j^{2},\ldots, S_j^{\gamma_{\veps/2}}$ such that for all $1 \leq i \leq \gamma_{\veps/2}, dia(S_j^{i}) \leq (\veps/2) \cdot (2^j R)$ and $\cup_{i=1}^{\gamma_{\veps/2}} S_j^{i}= Y_j$. 
Let $S_j$ be the set of points constructed by picking one point from each of the sets $S_j^{1},\ldots, S_j^{\gamma_{\veps/2}}$ and let $S = \cup_{j=0}^{t} S_j$. We have,
\begin{eqnarray*}
&~~&\Phi(S, Y)\\
&\leq& \sum_{j=0}^{t} \Phi(S_j, Y_j) \leq \sum_{j=0}^{t} {\veps}/{2} \cdot (2^j R) \cdot |Y_j| \\
&=& \sum_{j=0}^{t} \veps (2^{j-1} R) \cdot |Y_j| \\
&\leq& \sum_{j=0}^{t} \veps \cdot \Phi(\{c\}, Y_j) \leq \veps \cdot \Phi(\{c\}, Y) \leq \veps \cdot \Delta(Y,1)
\end{eqnarray*}

The size of the set $S$ is given by $|S| = \gamma_{\veps/2} \cdot t = O(\gamma_{\veps/2} \cdot \log{|Y|})$.
Generalizing this for all sets $X_1,\ldots,X_k$, we get that there is a set $\xi \subseteq \X$ of size $O(k \cdot \gamma_{\veps/2} \cdot \log{n})$ such that $\Phi(\xi, X) \leq \veps \cdot \Delta(X,k)$. The theorem follows from the bound on $\gamma_{\veps/2}$ from Lemma~\ref{lem:met4}.
\end{proof}

The lower bound below follows trivially from the lower bound for the Euclidean $k$-median problem. %(since $\R^d$ has a doubling dimension $d$).

\begin{theorem}[Lower bound for metric $k$-median]
For any $0 < \veps \leq 1/8$, any integer $k\geq 1$, there exists a metric space $(\X, D)$ with doubling dimension $d$ and a dataset $X \subseteq \X$ with $n$ points, such that for any set $\xi \subseteq \X$ with $\Phi(\xi, X) \leq \veps \cdot \Delta(X,k)$, $\xi$ is of size $\Omega \left( \frac{k \cdot \log{n}}{(16\veps)^{d-1}}\right)$.
\end{theorem}

\section{Coresets for $k$-means} \label{sec:coresets}

We start by recalling the notion of a movement-based coreset, defined as follows.

\begin{definition}[$(k, \veps)$-movement-based coreset]
A $(k, \veps)$-movement-based coreset of a dataset $X \subseteq \R^d$ is a set of points $S \subseteq \R^d$ such that $\Phi(S, X) \leq \veps \cdot \Delta(X,k)$.

\end{definition}

We will now see that movement-based coreset is a stronger notion than coreset in the sense that if any dataset $X$ has a $(k, \veps^{O(1)})$-movement-based coreset of size $m$, then it also has a $(k, \veps)$-coreset of size $m$.
A $(k, \veps)$-coreset is defined by a set $S$ and a weight function $w$, whereas a movement-based coreset is defined by just a set of points.
We will show that for any $(k, \veps)$-movement based coreset $S$, this set along with an appropriately defined weight function (which is dependent  just on $S$ and $X$) is also a $(k, \veps)$-coreset.
Let $X$ and $Y$ be any set of points.
For any point $y \in Y$, we define $\N_X(y)$ to be the set of points from $X$ such that their closest point in set $Y$ is $y$. That is, $\N_X(y) = \{x | x \in X \textrm{ and } \arg\min_{p \in Y} \norm{x - p} = y\}$. We give the proof of the result that a $(k,\veps)$-movement-based coreset implies a $(k,\veps)$-coreset in Appendix \ref{sec:appen2}. We note that this result was implicitly present in \cite{streamkm}.

We have argued that the notion of a movement-based coreset is a stronger notion than a coreset. As far as the existence of such movement-based coresets are concerned, we know from previous discussions that for any dataset $X$, there exists $(k, \veps)$-movement-based coreset of size $\Lp{\veps}$. Moreover, such movement-based coresets may be computed by running any constant factor pseudo-approximation algorithm for $k$-means with $k$ set as $\Lp{O(\veps)}$. As seen in the previous subsection, the $D^2$-sampling algorithm is one such algorithm. Combining Theorem~\ref{thm:pseudo} with Theorem~\ref{thm:coreset-1} given in the Appendix, we get the following result.

\begin{theorem} There is a universal constant $\beta$ for which the following holds: For any $0 < \veps \leq 1$, positive integer $k$ and any dataset $X \subseteq \R^d$, let $S$ denote the a set of centers chosen with $D^2$-sampling from $X$ such that $|S| = \Omega(\Lp{\veps^2/\beta})$. Then, $S$ is a $(k, \veps)$-coreset with high probability.
\end{theorem}

\section{Estimation of Intrinsic dimension} \label{sec:intrinsic}

We consider the following three notions for estimation of intrinsic dimension of a dataset.

\begin{enumerate}
\item {\it Affine subspace}: Suppose a dataset $X \subseteq \mathbb{R}^D$ sits in an affine subspace of dimension $d < D$. Theorem \ref{thm:main1} suggests that $\Lp{\veps}$ has a dependency of the form $({1}/{\veps})^{d/2}$ over $\veps$. So, if $m$ is the slope of the log-log plot of $\Delta(X,k)$ versus $k$, then $(-2.0)/m$ may be a good estimate of the dimension $d$. Please note that we are making a number of heuristic adjustments while making the previous statement. This estimation technique is a heuristic and Theorem 1 only provides a high-level justification for this heuristic. The only way to have more confidence would be to test the method on real data. This is also true for the two cases discussed below.

\item {\it Doubling Dimension}: Same as in the previous case but now using Theorem \ref{thm:metric-median} instead of Theorem \ref{thm:main1}, we get that $(-1.0)/m$ should provide a good estimate for the doubling dimension of the dataset $X$ where $m$ is the slope of the log-log plot of $\Delta(X,k)$ versus $k$. Here, $\Delta(X,k)$ denotes the $k$-median cost.

\item {\it Covariance Dimension}: Suppose a dataset $X \subseteq \mathbb{R}^D$ does not precisely belong to a $d$-dimensional affine subspace as in the first case but is close to one. This can be made more precise in the following manner: Suppose for simplicity that $X$ is centered around the origin. For an error parameter $\gamma$, the dataset $X$ is said to have a covariance dimension $d_{\gamma}$ if the sum of the first $d_{\gamma}$ eigenvalues of the covariance matrix of $X$ is at least $(1-\gamma)$ times the sum of all eigenvalues. This translates to the fact that the projection of all points to the space orthogonal to the space spanned by the first $d_{\gamma}$ eigenvectors has a very small contribution to the 1-means cost (only a $\gamma$ fraction).
So, as in the first case, $\Lp{\veps}$ will have dependency of the form $({1}/{\veps})^{d_{\veps}/2}$ over $\veps$. So, $(-2.0)/m$ should provide a good estimate for the covariance dimension of $X$ where $m$ is the slope of the log-log plot of $\Delta({X},k)$ versus $k$.
\end{enumerate}

So, our technique for dimension estimation basically involves estimating the slope of the log-log plot of $k$ versus $\Delta(X,k)$.
The main issue in obtaining such a plot is computing $\Delta(X,k)$ which is the optimal $k$-means cost for dataset $X$.
This is because $k$-means is an $\mathsf{NP}$-hard problem when $k > 1$.
There are a number of good approximation algorithms for the $k$-means problem and we can get an approximate value of $\Delta(X,k)$ for all values of $k$.
This leads us to an efficiency issue.
The issue is that we will need to repeatedly run the approximation algorithm for different values of $k$.
The $k$-means++ seeding algorithm is an interesting approximation algorithm for this context and provides a useful solution for this issue.
Note that the $k$-means++ seeding algorithm uses input $k$ only as a termination condition.
That is, it continues to sample centers using $D^2$-sampling as long as the number of centers is less than $k$.
The approximation analysis says that for any fixed value of $i$, the first $i$ centers sampled by the algorithm provides $O(\log{i})$ approximation to the $i$-means objective (in expectation).
So, if $C_i$ denotes the first $i$ centers sampled by the algorithm, then the heuristic we propose simply plots $i$ versus $\Phi(C_i, X)$.
Here is the pseudocode that we use to estimate the dimension of the given dataset $X$ when the cost function is sum of squared distances.\footnote{When the cost function is sum of distance as in $k$-median, the last line should return $(-1.0)/m$.} Let $\ell$ be an input parameter that gives better estimates when it is large but at the cost of the running time. Since the datasets on which we perform our experiments have small intrinsic dimension, setting $\ell = 100$ suffices.

\begin{framed}
{\tt DimensionEstimate($X, \ell$)}\\
\hspace*{0.1in} - $C_0 \leftarrow \{\}$\\
\hspace*{0.1in} - for $i$ = $1$:$\ell$\\
\hspace*{0.3in} - $D^2$-sample a center $c$ with respect to $C_{i-1}$ \\
\hspace*{0.3in} - $C_i \leftarrow C_{i-1} \cup \{c\}$\\
\hspace*{0.3in} - $Cost[i] \leftarrow \Phi(C_i, X)$\\
\hspace*{0.1in} - Let  $m$ be the slope of the best fit line for the curve \\
\hspace*{0.2in} $Cost[i]$ versus $i$ on a log-log plot\\
\hspace*{0.1in} - return($(-2.0)/m$)
\end{framed}

We now evaluate our heuristic by testing it on synthetic and real datasets.
We first perform experiment on a synthetic dataset with extrinsic dimension $100$ where the points are within a small $d$-dimensional affine subspace with $d$ ranging from $2$ to $7$. We call these six datasets {\tt Affine-2} through {\tt Affine-7}.
More specifically, here is how this dataset is generated:
First, we pick a random point $\mathbf{x} \in \mathbb{R}^{100}$.
We then randomly generate $100000$ points by adding gaussian noise (using $N(0, 1)$) independently to the first $d$ coordinates of $\mathsf{x}$.
Figure \ref{fig:1} shows the plot of $Cost[i]$ versus $i$ and the log-log plot for the same.
We note that the log-log plot are nearly straight lines and $(-2.0)/slope$ gives a close estimate for the intrinsic dimension.
For estimating the dimension, we repeat $10$ times and then report the average value.
The dimension estimates for {\tt Affine-2} to {\tt Affine-7} are $\{2.17, 3.27, 3.87, 4.65, 5.33, 6.26\}$.

\begin{figure}
    \centering
    \subfloat[Affine]{{\includegraphics[width=3.5cm]{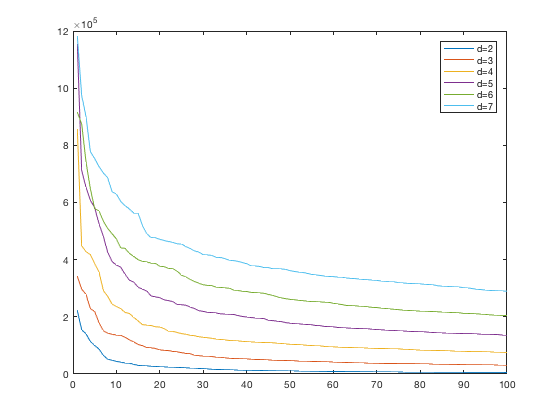} }}%
    \subfloat[Affine]{{\includegraphics[width=3.5cm]{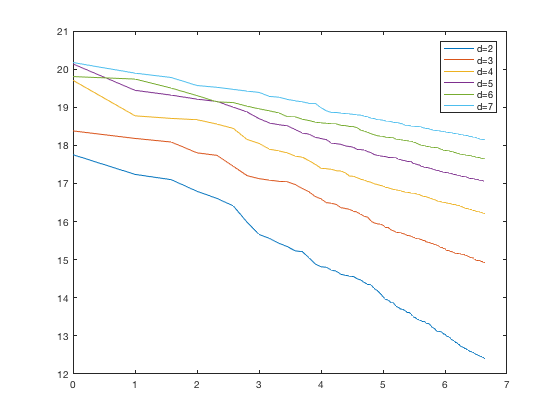} }}\\%
    \subfloat[Swissroll]{{\includegraphics[width=3.5cm]{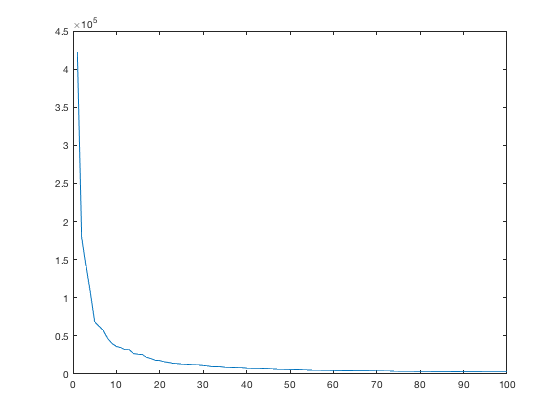} }}%    
    \subfloat[Swissroll]{{\includegraphics[width=3.5cm]{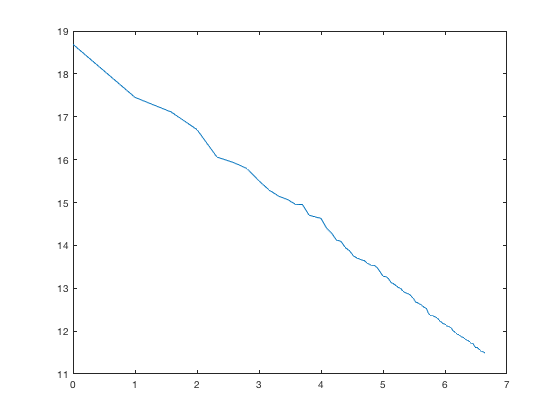} }}%    
    \caption{Figures (a) and (c) show plot of $Cost[i]$ vs $i$ and (b) and (d) show the corresponding log-log plot for the Affine and Swissroll dataset.}%
    \label{fig:1}
\end{figure}

We also experiment with some of the widely used datasets for dimension estimation.
The first of these datasets is the "swissroll" dataset \footnote{\url{http://people.cs.uchicago.edu/~dinoj/manifold/swissroll.html}} comprising of $1600$ points in 3-dimensions located on a manifold shaped like a {\em swissroll}.
So, the intrinsic dimension should be $2$.
Figure \ref{fig:1} gives the plots for this dataset.
The estimate of the dimension given by our technique is $1.8$.

Another dataset used is a dataset generated from a video of a rotating teapot.
One frame of this video is considered a data point and there are 100 frames.
Since there is only one degree of freedom of motion for the teapot, ideally the dimension of this dataset should be $1$.
However, the dimension estimate using our technique is $6.9$ suggesting that lighting, reflection and other parameters may also be important parameters that are not ignored.
Another interesting property of the log-log plot for the Teapots dataset that one should note is that the curve is not a straight line as is the case with the previous examples. This suggests that the dimension may depend on the scale at which the data is analysed.
At a much more finer scale, the data may have much larger dimension than at a coarser level.
The interesting property of our technique is that using the slope at various points in the curve may give dimension at various scales.
What is interesting about this observation is that for complex datasets where we have no intuition regarding the intrinsic dimension, the shape of the log-log plot may provide useful insight about the dimensionality of the data at various scales.
Another dataset that is similar to the teapots dataset is the hand dataset where the data points are $481$ frames of a video capturing a rotating hand holding a cup.
The dimension estimate using our technique is $2.6$.
This matches estimates using other techniques.
In particular, it matches the estimates of \cite{kegl}.

 \begin{figure}[htbp]
    \centering
    \subfloat[Teapots]{{\includegraphics[width=3.5cm]{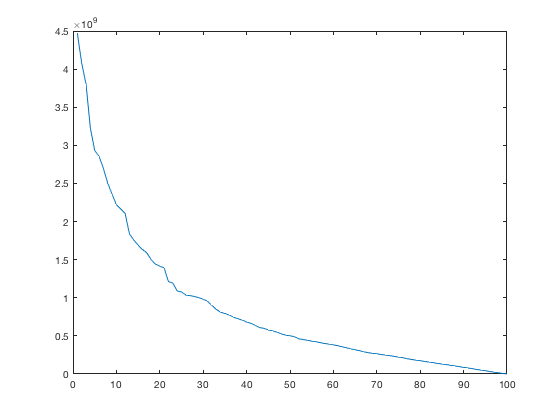}}}%
    \subfloat[Teapots]{{\includegraphics[width=3.5cm]{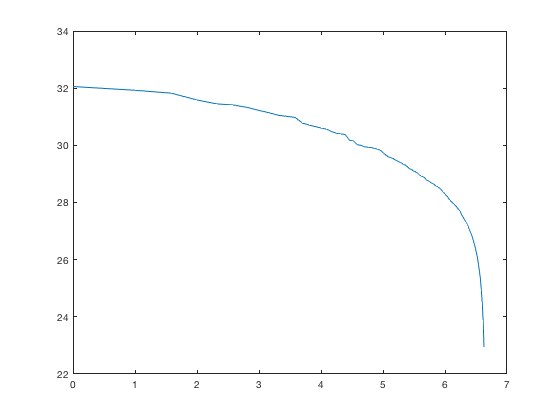} }}\\%
    \subfloat[Hand]{{\includegraphics[width=3.5cm]{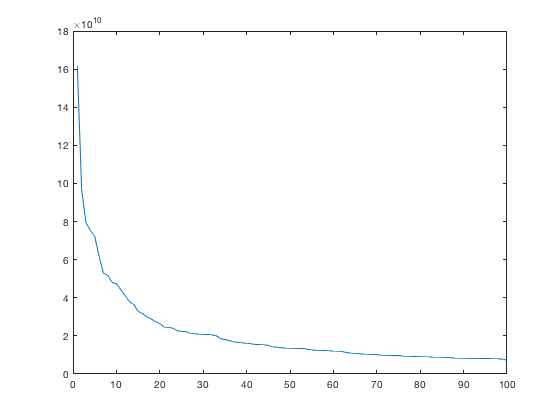}}}%
    \subfloat[Hand]{{\includegraphics[width=3.5cm]{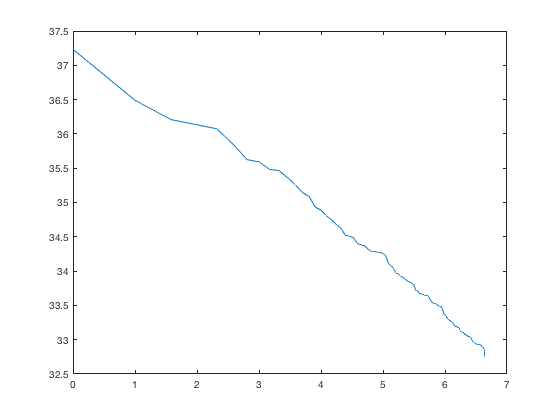} }}%
    \caption{Figures (a) and (c) show plot of $Cost[i]$ vs $i$ and (b) and (d) show the corresponding log-log plot for the Teapots and Hand dataset.}%
    \label{fig:2}%
\end{figure}

We perform experiments on a dataset where it is not clear what the intrinsic dimension might be.
We test our technique on the MNIST test dataset for hand written number "2".
The number of data points is around $1000$ and the extrinsic dimension is $784$ ($28 \times 28$ pixels).
Our estimate for the intrinsic dimension is $11.4$.
\footnote{However, one should mention that the estimate has a high variance with the standard deviation being $2.6$ over 30 runs.}
This dimension estimate closely matches the estimate of a number of previous works. The estimate in many previous works are in the range $12$-$14$ (see e.g., \cite{gc16}).

 \begin{figure}[htbp!]
    \centering
    \subfloat[\url{http://vasc.ri.cmu.edu/idb/html/motion/hand/index.html}]{{\includegraphics[scale=0.2]{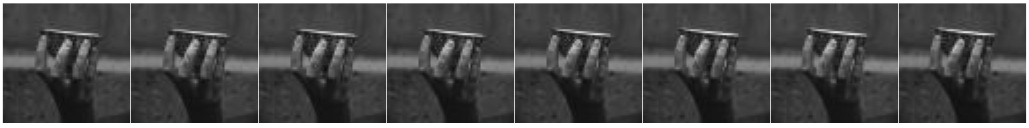}}}\\%
    \subfloat[\url{http://cseweb.ucsd.edu/~saul/papers/sde_cvpr04.pdf}]{{\includegraphics[scale=0.45]{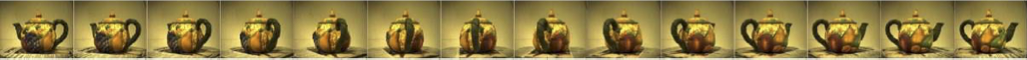} }}\\%
    \subfloat[\url{http://yann.lecun.com/exdb/mnist/}]{{\includegraphics[scale=0.3]{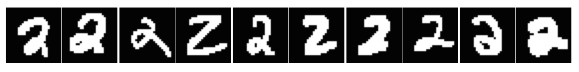} }}\\%        
    \caption{Hand, Teapot, and MNIST(2) datasets}%
    \label{fig:3}%
\end{figure}

\bibliographystyle{plain}
\bibliography{ref}

\appendix
\section{Bounds on $\N(\Sp, \veps)$ and $\Pa(\Sp, \veps)$} \label{sec:appen}

We restate Lemma~\ref{lemma:geometric-bounds} below before giving the proof.

\begin{lemma}[Restatement of Lemma~\ref{lemma:geometric-bounds}]
Let $\Sp$ denote a unit sphere in $\R^d$ and $0 < \veps < 1$. 
Then
\begin{enumerate}
\item $\N(\Sp, \veps) = O\left(\frac{1}{(\veps/8)^{d-1}}\right)$, and
\item $\Pa(\Sp, \veps) = \Omega\left(\frac{1}{(2\veps)^{d-1}}\right)$.
\end{enumerate}
\end{lemma}

\begin{proof}
We will use the following well known fact to obtain our bounds:
\begin{equation}\label{fact:cover-pack}
    \Pa(\Sp, 2\veps) \leq \N(\Sp, \veps) \leq \Pa(\Sp, \veps)
\end{equation}
%The upper bound on $\N(\Sp, \veps)$ is based on a simple volume argument. Let $S$ be any $\veps$-covering set over the unit sphere with minimum cardinality. So, $|S| = \N(\Sp, \veps)$. Consider balls of radius $\veps/2$ incident on the points of $S$. Using the fact that $S$ is of minimum cardinality cover, we get that these balls do not intersect. Note that all these balls lie inside a ball of radius $(1 + \veps/2)$ centered at the origin. The total volume of these balls is given by $|S| \cdot V(\veps/2)$ and volume of the ball of radius $(1+\veps/2)$ is given by $V(1+\veps/2)$. Here $V(r)$ denote the volume of a ball of radius $r$ and is proportional to $r^d$. So, we have: $|S| \cdot V(\veps/2) \leq V(1+\veps/2)$ which implies that $|S| = \N(\Sp, \veps) \leq \frac{V(1+\veps/2)}{V(\veps/2)} = (1+2/\veps)^d$.

%Let $P$ be any $\veps$-packing set over unit sphere of maximum cardinality. 
Let $p$ be any point in $\Sp$ and let $S(\gamma)$ be the spherical cap over the unit sphere formed by taking the intersection of a ball of radius $\gamma$ at $p$ with the surface of the unit sphere.
We will now upper bound the surface area of the spherical cap $S(\gamma)$.
Let $A_d(r)$ denote the surface area of a sphere of radius $r$ in a $d$-dimensional Euclidean space. 
The bound on the surface area can be calculated using the following integral (see Figure~\ref{fig:bounds} for reference):
\begin{figure}
\centering
\includegraphics[scale=0.18]{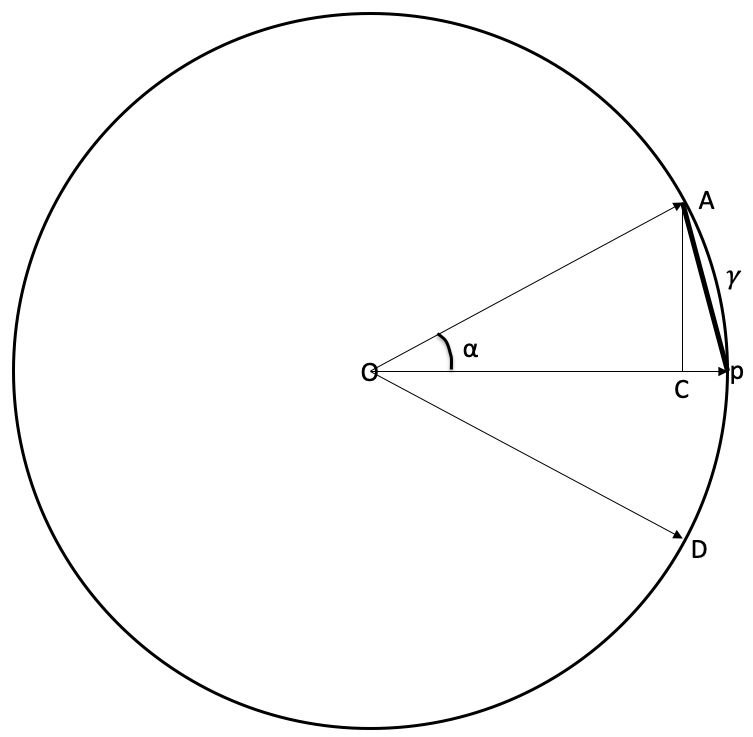}
\caption{We consider sphere of radius $\gamma$ centered on $p$. $\alpha$ is the angle of the spherical cap at the origin.}
\label{fig:bounds}
\end{figure}
\begin{eqnarray*}
&~~&S(\gamma)\\
&=& \int_{0}^{\alpha} A_{d-1}(\sin{\theta}) d \theta \\
&=& A_{d-1}(1) \cdot \int_{0}^{\alpha} (\sin{\theta})^{d-2} d\theta \\
&\leq& A_{d-1}(1) \cdot \int_{0}^{\alpha} \theta^{d-2} d \theta = A_{d-1}(1) \cdot \frac{\alpha^{d-1}}{d-1}
\end{eqnarray*}
Similarly, we can also obtain an upper bound on $S(\gamma)$:
\begin{eqnarray*}
S(\gamma) &=& \int_{0}^{\alpha} A_{d-1}(\sin{\theta}) d \theta \\
&=& A_{d-1}(1) \cdot \int_{0}^{\alpha} (\sin{\theta})^{d-2} d\theta \\
&\geq& A_{d-1}(1) \cdot \int_{0}^{\alpha} \frac{\theta^{d-2}}{2^{d-2}} d \theta  \\
&=& A_{d-1}(1) \cdot \frac{(\alpha/2)^{d-1}}{d-1}
\end{eqnarray*}
Using the fact that $\sin(\alpha) = \gamma\sqrt{1 - \gamma^2/4}$, we obtain the following bounds on $S(\veps)$ and $S(\veps/2)$ that we will use.
\[
S(\veps) \leq A_{d-1}(1)\cdot \frac{(2\veps)^{d-1}}{d-1} \ \ \textrm{and}\ \ S\left(\frac{\veps}{2}\right) \geq A_{d-1}(1) \cdot \frac{(\frac{\veps}{8})^{d-1}}{d-1}
\]
Let $N$ be any $\veps$-covering set over the unit sphere with minimal cardinality. Then we have $A_{d}(1) \leq |N| \cdot S(\veps)$ which gives:
\begin{eqnarray*}
|N| \geq \frac{A_d(1)}{A_{d-1}(1)} \cdot \frac{d-1}{(2\veps)^{d-1}} = \frac{2 \pi (d-1)}{d} \cdot \frac{1}{(2\veps)^{d-1}}.
\end{eqnarray*}
Using eqn.(\ref{fact:cover-pack}), we get that $\Pa(\Sp, \veps) = \Omega\left(\frac{1}{(2\veps)^{d-1}}\right)$.

Let $P$ be an $\veps$-packing set of maximum cardinality over the unit sphere. 
Note that balls of radius $(\veps/2)$ centered at points in $P$ do not intersect. 
This gives $|P| \cdot S(\veps/2) \leq A_{d}(1)$. 
So, using lower bound on the surface area of the spherical cap, we obtain:
\begin{eqnarray*}
|P| \leq \frac{A_d(1)}{A_{d-1}(1)} \cdot \frac{d-1}{(\veps/8)^{d-1}} = \frac{2 \pi (d-1)}{d} \cdot \frac{1}{(\veps/8)^{d-1}}.
\end{eqnarray*}
Using eqn.(\ref{fact:cover-pack}), we get that $\N(\Sp, \veps) = O\left(\frac{1}{(\veps/8)^{d-1}}\right)$.
\end{proof}

\section{$(\lowercase{k},\veps)$-movement-based Coreset implies a $(\lowercase{k},\veps)$-coreset} \label{sec:appen2}

\begin{theorem}\label{thm:coreset-1}
Let $X \subseteq \R^d$ be any dataset and $S$ be a $(k, \veps^2/32)$-movement-based coreset of $X$. Let $w: S \rightarrow \R^{+}$ be a weight function defined as follows $\forall s \in S, w(s) = |\N_X(s)|$, where for any point $s \in S$, we define $\N_X(s)$ to be the set of points from $X$ such that their closest point $S$ is $s$. Then, $S$ along with weight function $w$ is a $(k,\veps)$-coreset of $X$. \end{theorem}

\begin{proof}
Let $C$ be any set of $k$ centers. For any point $x \in X$, let $c_x$ denote its closest point in the set $C$. Similarly, let $s_x$ denote its closest point in the set $S$. Given this, first we note that for any point $x \in X$, if $D(x, c_x) > D(s_x, c_{s_x})$ we have:
\begin{eqnarray*}
|D(x, c_x) - D(s_x, c_{s_x})|
&\leq& D(x, c_{s_x}) - D(s_x, c_{s_x}) \\
&\leq& D(x, s_x),
\end{eqnarray*}
and when $D(x, c_x) \leq D(s_x, c_{s_x})$, we have:
\begin{eqnarray*}
|D(x, c_x) - D(s_x, c_{s_x})| &\leq& D(s_x, c_x) - D(x, c_x)\\
&\leq& D(x, s_x)
\end{eqnarray*}
So, from the above two inequalities, we get $|D(x, c_x) - D(s_x, c_{s_x})| \leq D(x, s_x)$.
Now, for every point $x \in X$, $|D(x, c_x)^2 - D(s_x, c_{s_x})^2|$ equals:
\begin{eqnarray*}
&&|D(x, c_x)^2 - D(s_x, c_{s_x})^2|\\
&=& |D(x, c_x) - D(s_x, c_{s_x})| \cdot (D(x, c_x) + D(s_x, c_{s_x})) \\
&\leq& D(x, s_x) \cdot (D(x, c_x) + D(s_x, c_{x}))  \\
&& \textrm{(since $D(s_x, c_{s_x}) \leq D(s_x, c_x)$ and} \\
&& \textrm{$|D(x, c_x) - D(s_x, c_{s_x})| \leq D(x, s_x)$)} \\
&\leq& D(x, s_x) \cdot (2 \cdot D(x, c_x) + D(x, s_x))\\
&& \textrm{(since $D(s_x, c_x) \leq D(x, s_x) + D(x, c_x)$)}\\
&=& 2 \cdot D(x, s_x) \cdot D(x, c_x) + D(x, s_x)^2
\end{eqnarray*}
We will use the above inequality to now show the main result. 
$|\Phi(C, S, w) - \Phi(C, X)|$ can be upper bounded by:
\begin{eqnarray*}
&& \sum_{x \in X} |D(x, c_x)^2 - D(s_x, c_{s_x})^2|\\
&\leq& \sum_{x \in X} 2 \cdot D(x, s_x) \cdot D(x, c_x) + \sum_{x \in X} D(x, s_x)^2 \\
&=& 2 \cdot \sum_{x \in X, \frac{4}{\veps} \cdot D(x, s_x) <  D(x, c_x)} D(x, s_x) \cdot D(x, c_x) + \\
&& 2 \cdot \sum_{x \in X, \frac{4}{\veps} \cdot D(x, s_x) \geq  D(x, c_x)} D(x, s_x) \cdot D(x, c_x)+\\
&& \sum_{x \in X} D(x, s_x)^2\\
&\leq & 2 \cdot \sum_{x \in X, \frac{4}{\veps} \cdot D(x, s_x) <  D(x, c_x)} (\veps/4) \cdot D(x, c_x)^2 + \\
&& 2 \cdot \sum_{x \in X, \frac{4}{\veps} \cdot D(x, s_x) \geq  D(x, c_x)} \frac{4}{\veps} \cdot D(x, s_x)^2+\\
&& \sum_{x \in X} D(x, s_x)^2\\
&\leq & \frac{\veps}{2} \cdot \sum_{x \in X}  D(x, c_x)^2 + \frac{8}{\veps} \cdot \sum_{x \in X}  D(x, s_x)^2+\\
&& \sum_{x \in X} D(x, s_x)^2\\
&=& \frac{\veps}{2} \cdot \Phi(C, X) + (8/\veps + 1) \cdot \Phi(S, X)\\
&\leq& \frac{\veps}{2} \cdot \Phi(C, X) + \frac{\veps}{2} \cdot \Delta_X(k)\\
&\leq& \veps \cdot \Phi(C, X)
\end{eqnarray*}
This implies that $|\Phi(C, S, w) - \Phi(C, X)| \leq \veps \cdot \Phi(C, X)$, which in turn means that $S$ is a $(k, \veps)$-coreset of $X$. This completes the proof of the theorem.
\end{proof}

\end{document}